\documentclass[review]{elsarticle}
\usepackage{mathtools}
\usepackage{tikz}
\usepackage{amsmath,physics}
\usetikzlibrary{arrows}
\usepackage{cleveref}
\usepackage[normalem]{ulem}
\usepackage[T1]{fontenc}
\usepackage{amsfonts}
\usepackage[plain]{algorithm}
\usepackage{algpseudocode}

\newcommand{\ind}[1]{\mbox{#1}}

\newtheorem{theorem}{Theorem}
\newtheorem{lemma}{Lemma}

\newtheorem{remark}{Remark}

\newproof{proof}{Proof}

\begin{document}
\begin{frontmatter}

\title{Correcting matrix products over the ring of integers}

\author[hlwang]{Yu-Lun~Wu}
\ead{60947037s@ntnu.edu.tw}

\author[hlwang]{Hung-Lung~Wang\corref{cor1}}
\ead{hlwang@ntnu.edu.tw}

\address[hlwang]{Department of Computer Science and Information Engineering, \\
National Taiwan Normal University, \\
No.88, Sec. 4, Tingzhou Rd., Wenshan Dist., Taipei City 116, Taiwan}

\cortext[cor1]{Corresponding author}

\begin{abstract}
 Let $A$, $B$, and $C$ be three $n\times n$ matrices. We investigate the problem of verifying whether $AB=C$ over the ring of integers and finding the correct product $AB$. Given that $C$ is different from $AB$ by at most $k$ entries, we propose an algorithm that uses $O(\sqrt{k}n^2+k^2n)$ operations. 
    Let $\alpha$ be the largest absolute value of an entry in $A$, $B$, and $C$. The integers involved in the computation are of $O(n^3\alpha^2)$.
\end{abstract}

\begin{keyword}
matrix multiplication\sep certifying algorithm\sep correction
\end{keyword}

\end{frontmatter}

\section{Introduction}

Designing efficient algorithms to multiply two matrices is an ongoing research topic, due to its wide applications. The research originates from Strassen~\cite{Strassen69} in 1969, and the known fastest algorithm was proposed by Duan, Wu, and Zhou~\cite{DuanWZ23} in 2023, with time complexity $O(n^{2.37188})$, where $n$ is the row- and column-dimensions of the matrices. Although the best known algorithm for matrix multiplication is still super-quadratic, verifying  whether the product of two matrices is equal to a given matrix can be done more efficiently. Freivalds~\cite{Freivalds77} proposed a randomized algorithm that verifies with high probability whether $AB=C$, for given matrices $A$, $B$, and $C$. In Freivalds' algorithm, a vector $\vectorbold{v}$ is chosen at random, and the result of whether $A(B\vectorbold{v})=C\vectorbold{v}$ is reported. Clearly, Freivalds' algorithm has one-sided error; namely the answer may be false-positive. Nevertheless, there are only $O(n^2)$ operations needed. 

\medskip
It is still unknown whether Freivalds' algorithm can be efficiently derandomized~\cite{GkasieniecLLPT17}. An immediate attempt for derandomizing Freivalds' algorithm is to choose a vector $\vectorbold{x}$ of the form $(1, x, x^2, \dots, x^n)$. Then $AB\vectorbold{x}-C\vectorbold{x}$ consists of $n$ polynomials in $x$. The problem reduces to determining whether the nonzero polynomials are all zero at ${x}$. If $x$ is chosen as a non-root to all nonzero polynomials, then the modified algorithm verifies whether $AB=C$ deterministically. However, even though such an $x$ can be found, for example, by applying Cauchy's bound, one needs to manipulate numbers which are exponentially large. This makes the algorithm impractical. Related discussion can be found in~\cite{GkasieniecLLPT17,IvJi14}. 

\medskip
Recently, G{\k{a}}sieniec et al.~\cite{GkasieniecLLPT17} showed that, if the matrix $C$ has at most $k$ erroneous entries, then these entries can be identified and also corrected in deterministic  $\tilde{O}(kn^2)$ time, where $\tilde{O}(\cdot)$ suppresses polylogarithmic terms in $n$ and $k$. Their approach resembles methods from combinatorial group testing~\cite{DuHwang99}. Based on the same concept, randomized algorithms were also developed for $k$ equal to the number of erroneous entries. When the number of erroneous entries is not known, they applied the technique of compressed matrix multiplication~\cite{Pagh13}, and developed a randomized algorithm which runs in $\tilde{O}(n^2+kn)$ time. Note that the algorithms of G{\k{a}}sieniec et al. applies for matrix multiplications over any ring. For matrix multiplications over the ring of integers, Kutzkov~\cite{kutzkov} developed a deterministic algorithm running in $O(n^2 +k^2 n \log^5 n)$ time; K\"{u}nnemann~\cite{Kunnemann18} proposed an ${O}(\sqrt{k}n^2\log^{2+o(1)} n+k^2\log^{3+o(1)}n)$-time deterministic algorithm, given that the absolute values of the input integers are upper bounded by $n^c$ for some constant $c$.

\medskip
In this paper, we are also concerned with the problem of correcting $C$ to be the product of $A$ and $B$, where $A$, $B$, and $C$ are three $n\times n$ matrices with $C$ different from $AB$ by at most $k$ entries. The matrix multiplication is restricted to be over the ring of integers. We develop a deterministic algorithm, which runs in $O(\sqrt{k}n^2+k^2n)$ time. 
Our algorithm is purely combinatorial; i.e. it does not rely on fast matrix multiplications or other well-developed subroutines. The idea is inherited from  G{\k{a}}sieniec et al.~\cite{GkasieniecLLPT17}, based on combinatorial group testing. This can also be viewed as an extension of Freivalds' algorithm; a number of vectors $\vectorbold{x}_1, \vectorbold{x}_2, \dots, \vectorbold{x}_m$ are chosen, and the erroneous entries are identified and corrected based on the result of $(AB-C)(\vectorbold{x}_1\mid  \vectorbold{x}_2\mid  \dots\mid \vectorbold{x}_m)$. 

\medskip
We note here that the values manipulated by the algorithm are of $O(\alpha^2n^3)$, where $\alpha$ is the largest absolute value of an entry in the input matrices. All values to manipulate are in a reasonable range, namely polynomial in both $n$ and the input values. In the remainder of this paper, 
we summarize necessary tools in Section~\ref{sec:pre}, and then give the algorithm in Section~\ref{sec:main}.




\section{Preliminaries}
\label{sec:pre}

For a matrix $A$, we use $A_{ij}$ to denote the entry at the intersection of the $i$th row and the $j$th column. We use $[n]$ to denote the set $\{x\in\mathbb{N}\colon\, 1\leq x\leq n\}$. 
Let $x_1,x_2,\dots, x_n$ be $n$ integers. The $n\times n$ Vandermonde matrix of $x_1,x_2,\dots, x_n$ is defined to be the matrix $V$ with $V_{ij}=x_i^{j-1}$; i.e.

\[V=\left(\begin{array}{ccccc}
    1 & x_1 & x_1^2 & \dots & x_1^{n-1} \\
    1 & x_2 & x_2^2 & \dots & x_2^{n-1} \\
    \vdots & \vdots & \vdots & \ddots & \vdots \\
    1 & x_n & x_n^2 & \dots & x_n^{n-1} \\ 
\end{array}\right)\]

\medskip
\noindent
The determinant of the $n\times n$ Vandermonde matrix is 
\begin{equation}
\label{eq:det_vand}
    \prod_{1\leq i<j\leq n}(x_j-x_i).
\end{equation}

\noindent
Let $p$ be a prime greater than $n$, and let $x_1, x_2,\dots, x_n$ be $n$ distinct integers less than $p$. Consider the matrix  $X$, with $0\leq X_{ij}<p$ and 
\[X_{ij}\equiv x_i^{j-1}\pmod{p}.\]

\noindent
By~\eqref{eq:det_vand}
\[\det(X)\equiv \prod_{1\leq i<j\leq n}(x_j-x_i)\pmod{p}.\]
Since $p$ is prime, $\det(X)\neq 0$. Thus, $X$ is invertible over $\mathbb{R}$. The following lemma is an immediate consequence of the discussion above.

\begin{lemma}
\label{lem:key}
    Let $p$ be a prime, and let $x_1, x_2,\dots, x_n$ be $n$ distinct integers less than $p$. For the $n\times n$ matrix $X$, with $0\leq X_{ij}<p$ and $X_{ij}\equiv x_i^{j-1} \pmod{p}$, and a column vector  $\vectorbold{y}\in\mathbb{R}^n$,  
    \[X\vectorbold{y} = \mathbf{0} \iff \vectorbold{y} = \mathbf{0}.\]
\end{lemma}

\medskip
\section{Correcting the product}
\label{sec:main}

In the following, we use $n$ to denote the row- and column-dimensions of the input matrices, and we let $p$ be a prime number greater than $n$. 



\subsection{The certificate} 
\label{sec:certificate}

Let $m$ be a positive integer, and let $V$ be an $n\times m$ matrix with entries in $[p]$ such that
\[V_{ij}\equiv i^{j-1} \pmod{p}. \]
We call $V$ an \textit{$m$-certificate}. 
Given the $m$-certificate $V$, a nonzero row $\vectorbold{u}$ of $AB-C$  is \textit{$V$-detectable} if $\vectorbold{u}V\neq \mathbf{0}$. Similarly,  a nonzero column $\vectorbold{u}$ of $AB-C$ is {$V$-detectable} if $\vectorbold{u}^TV\neq \mathbf{0}$.

\begin{lemma}
\label{lem:detectability}
    Let $V$ be an $m$-certificate. Then every nonzero row or column of $AB-C$ containing at most $m$ nonzeros is $V$-detectable. 
\end{lemma}

\begin{proof}
    Let $\vectorbold{u}=(u_1,\dots, u_n)$ be a nonzero row of $AB-C$, and let $S = \{i\in[n]\colon\, u_i\neq 0\}$. The submatrix $V_S$ of $V$ consisting of the rows indexed by $S$ is an $m\times m$ Vandermonde matrix of $m$ distinct elements over $\mathbb{Z}/p\mathbb{Z}$. By Lemma~\ref{lem:key}, $\vectorbold{u}\neq \mathbf{0}$ iff $\vectorbold{u}_SV_S\neq \mathbf{0}$, where $\vectorbold{u}_S$ is the subvector consisting of entries indexed by~$S$.
\qed\end{proof}

\medskip
In the following,  we assume that $\sqrt{k}$ is an integer to ease the presentation, and $m$ is fixed to be $\sqrt{k}$. For succinctness a $\sqrt{k}$-certificate $V$ is simply called a \textit{certificate}, and a row or a column that is $V$-detectable is simply said to be \textit{detectable}. 
For a certificate  $V$, the matrices $\ind{IC}=(AB-C)^TV$ and $\ind{IR}=(AB-C)V$ are called the \textit{column-indicator} and 
\textit{row-indicator} of $AB-C$, respectively. 

\begin{lemma} \label{lem:both_zero}
    Given three $n\times n$ matrices $A$, $B$ and $C$, let $V$ be a certificate, and let $\ind{IC}$ and $\ind{IR}$ be the \textit{column-indicator} and \textit{row-indicator} of $AB-C$, respectively.  If $\ind{IC}={\bf 0}$ and $\ind{IR}={\bf 0}$, then $AB=C$.
\end{lemma}

\begin{proof}
    Since $\ind{IC}={\bf0}$ and $\ind{IR}={\bf 0}$, by Lemma~\ref{lem:detectability}, every column or row of $AB-C$ that contains a nonzero has more than $\sqrt{k}$ nonzeros. The existence of a nonzero column implies that $AB-C$ contains more than $k$ nonzeros.   
\qed\end{proof}



Observe that for a row of $AB-C$ containing more than $\sqrt{k}$ nonzeros, one of these nonzero entries must be in a column containing less than $\sqrt{k}$ nonzeros, given that $AB-C$ contains at most $k$ nonzeros. As an immediate consequence, we note the following.

\begin{remark}\label{rmk:one_zero}
    Every non-detectable row (respectively, column) contains a nonzero entry residing in a detectable column (respectively, row). 
\end{remark}

We end this section by giving the pseudocode for computing the indicators. Note that the sets, $S$ and $T$, of indices of nonzero rows of $\ind{IR}$ and nonzero columns of $\ind{IC}$ are essential for our algorithm. In the following, we denote the indicators by $(\ind{IR},S)$ and $(\ind{IC},T)$. 

\medskip
\begin{algorithmic}[1]
\Function{RowIndicator}{$A$, $B$, $C$, $V$}
   \State $\ind{IR}\gets A(BV)-CV$
   \State $S\gets \{i\in[n]\colon\, \mbox{row $i$ of \ind{IR} is nonzero}\}$
   \State \textbf{return} $(\ind{IR},S)$
\EndFunction

\medskip
\Function{ColIndicator}{$A$, $B$, $C$, $V$}
   \State $\ind{IC}\gets (V^TA)B-V^TC$
   \State $T\gets \{j\in[n]\colon\, \mbox{column $j$ of \ind{IC} is nonzero}\}$
   \State \textbf{return} $(\ind{IC},T)$
\EndFunction

\end{algorithmic}

\subsection{The algorithm} 
\label{sec:algo}

\begin{algorithm}[t]
\begin{algorithmic}[1]

\Require $AB-C$ contains no more than $k$ nonzeros, $V$ is the certificate
\Ensure $C=AB$
\State $(\ind{IR}, S)\gets${\sc RowIndicator}$(A, B, C, V)$ 
\State $(\ind{IC}, T)\gets$ {\sc ColIndicator}$(A, B, C, V)$
\For{$(i,j)\in (S\times [n]) \cup ([n]\times T)$}
    \State $C_{ij}\gets (AB)_{ij}$ 
\EndFor 
\end{algorithmic}
\caption{An $O(kn^2)$-time algorithm for correcting the matrix product with no more than $k$ errors.}\label{alg:primary}
\end{algorithm}

Once we have the row- and column-indicators $(\ind{IR},S)$ and $(\ind{IC},T)$, all erroneous entries in $C$ can be corrected using $O(kn^2)$ operations (see Algorithm~\ref{alg:primary} below). If $\ind{IC}=\mathbf{0}$ and $\ind{IR}=\mathbf{0}$, then $AB=C$ (Lemma~\ref{lem:both_zero}). Otherwise, we identify nonzero rows and columns, and then recompute all these rows and columns one by one.
 As mentioned in Remark~\ref{rmk:one_zero}, the index $i$ of a nonzero row is either in $S$ or there exists $j\in T$ such that $M_{ij}\neq 0$. Therefore, all nonzero rows can be identified either via $S$, or via computing the columns indexed by $T$; a symmetric procedure can be applied to identify all the nonzero columns. Since there are no more than $k$ nonzero entries in $M$, at most $k$ rows and columns of $AB$ are identified and to be computed. A straightforward computation using $O(n^2)$ operations per row/column leads to an algorithm of $O(kn^2)$ operations.




\medskip

\begin{remark}
    The upper bound of $kn^2$ is tight for the algorithm described above. For instance, if each of the $i$th row and the $j$th column of $AB-C$ contains about $k/2$ nonzeros, then each nonzero entry at row $i$ or column $j$ is located using  $\Theta(n^2)$ operations. 
    
\[\left(\begin{array}{cccccc}
    0 & 0 & \dots & a_{1j} & \dots & 0 \\
    0 & 0 & \dots &a_{2j} & \dots & 0 \\
    \vdots &\vdots &  & \vdots &  & \vdots \\
    a_{i1} & a_{i2} & \dots &a_{ij} & \dots & a_{in} \\ 
        \vdots & \vdots & & \vdots &  & \vdots \\
            0 & 0 &\dots & a_{nj} & \dots & 0 \\
\end{array}\right)\]
\end{remark}

%
%

\medskip
Below, we modify Algorithm~\ref{alg:primary} to obtain an $O(k^2 n + \sqrt{k} n^2)$-time algorithm. The algorithm works in two phases. For clarity, let $(\ind{IR}_0, S_0)$ and $(\ind{IC}_0, T_0)$ be the indicators in the very beginning, and let $M_0=AB-C$ be the matrix before the correction on $C$ starts. In the first phase, entries residing in both a detectable row and a detectable column are corrected. Note that right after this phase, nonzero entries are in a row belonging to $[n]\setminus S_0$ or a column belonging to $[n]\setminus T_0$, which means that each of these rows and columns contains more than $\sqrt{k}$ nonzeros in $M_0$. Then, in the second phase the remaining nonzero rows and columns are identified and corrected as in Algorithm~\ref{alg:primary}, except that the indicators are kept up-to-date by the procedure {\sc Update}. An illustration is given in Figure~\ref{fig:algo2}, and the pseudocode of the modified algorithm is given in Algorithm~\ref{alg:main}.

\begin{figure}
\[
\underset{
\begin{array}{c}
\\
S_0\cap S=\{1,3,4,5,6\}\\
\end{array}
}{\left(
\begin{array}{cccccc}
      0   &   0   &   0   & \ast & 0  &0\\
    \ast & \ast & \ast  &  \ast  & \ast &0\\
      0   &  \ast   &   0   & \ast & 0 &0\\
      0   &  \ast   &   0   & \ast & 0 &0\\
      0   &  \ast   &   0   & \ast & 0 &0\\
      0   &  \ast   &   0   &  0 & 0 &0\\
\end{array}\right) }  
\implies 
\underset{
\begin{array}{c}
\\
S_0\cap S=\{1,3,4,5,6\}\\
\end{array}
}{\left(
\begin{array}{cccccc}
      0   &   0   &   0   & [\ast] & 0  &0\\
    \ast & \ast & \ast  &  \ast  & \ast &0\\
      0   &  \ast   &   0   & \ast & 0 &0\\
      0   &  \ast   &   0   & \ast & 0 &0\\
      0   &  \ast   &   0   & \ast & 0 &0\\
      0   &  \ast   &   0   &  0 & 0 &0\\
\end{array}\right) }  
\implies 
\underset{
\begin{array}{c}
\\
S_0\cap S=\{3,4,5,6\}\\
\end{array}
}{\left(
\begin{array}{cccccc}
      0   &   0   &   0   & \mathbf{0} & 0  &0\\
    \ast & \ast & \ast  &  \mathbf{0}  & \ast &0\\
      0   &  \ast   &   0   & \mathbf{0} & 0 &0\\
      0   &  \ast   &   0   & \mathbf{0} & 0 &0\\
      0   &  \ast   &   0   & \mathbf{0} & 0 &0\\
      0   &  \ast   &   0   &  0 & 0 &0\\
\end{array}\right) }  
\]

\bigskip
\[
\implies 
\underset{
\begin{array}{c}
\\
S_0\cap S=\{3,4,5,6\}\\
\end{array}
}{\left(
\begin{array}{cccccc}
      0   &   0   &   0   & {0} & 0  &0\\
    \ast & \ast & \ast  &  {0}  & \ast &0\\
      0   &  [\ast]   &   0   & {0} & 0 &0\\
      0   &  \ast   &   0   & {0} & 0 &0\\
      0   &  \ast   &   0   & {0} & 0 &0\\
      0   &  \ast   &   0   &  0 & 0 &0\\
\end{array}\right) }  
\implies 
\underset{
\begin{array}{c}
\\
S_0\cap S=\emptyset\\
\end{array}
}{\left(
\begin{array}{cccccc}
      0   &   0   &   0   & {0} & 0  &0\\
    \ast & \mathbf{0}  & \ast  &  {0}  & \ast &0\\
      0   &  \mathbf{0}   &   0   & {0} & 0 &0\\
      0   &  \mathbf{0}    &   0   & {0} & 0 &0\\
      0   &  \mathbf{0}    &   0   & {0} & 0 &0\\
      0   &  \mathbf{0}    &   0   &  0 & 0 &0\\
\end{array}\right) }  
\]

\caption{An illustration for the first half of Phase~2 (lines~8 to~15) of Algorithm~\ref{alg:main}, assuming $\sqrt{k}=4$. We use the matrix $C-AB$ to illustrate the process of correction. First, an element $i$ of $S_0\cap S$ is chosen (line~9), say $i=1$. Then, row $1$ is recomputed to identify the nonzero columns, and in this example only column $4$ is identified (line~10). All nonzero entries in column~$4$ are corrected, and $S_0\cap S$ is updated accordingly (lines~11 to~14). Notice that $S=\{2,3,4,5,6\}$ at the moment. The for loop stops after the second round. The remaining nonzero entries  will be corrected later in lines 16 to~23.}
\label{fig:algo2}
\end{figure}

\begin{algorithm}
\begin{algorithmic}[1]
\Function{Update}{$i, j, A, B, C, V, (\ind{IR},S), (\ind{IC},T)$}
   \State $x\gets (AB)_{ij}-C_{ij}$
   \For{$l\in \left[\sqrt{k}\;\right]$}   
        \State $\ind{IR}_{il}\gets \ind{IR}_{il} + x\cdot V_{jl}$
        \State $\ind{IC}_{lj}\gets \ind{IC}_{lj} + x\cdot V_{il}$
   \EndFor  
   \State Update $S$ and $T$ accordingly 
   \State \textbf{return} $(\ind{IR},S), (\ind{IC},T)$
\EndFunction
\end{algorithmic}
\end{algorithm}

\begin{algorithm}[H]
\caption{An $O(k^2 n + \sqrt{k} n^2)$-time algorithm for correcting the matrix product with no more than $k$ errors.}\label{alg:main}
\begin{algorithmic}[1]

\Require $AB-C$ contains no more than $k$ nonzeros, $V$ is the certificate
\Ensure $C=AB$
\State $(\ind{IR}_0, S_0)\gets${\sc RowIndicator}$(A, B, C, V)$ \label{proc:phase1_start}
\State $(\ind{IC}_0, T_0)\gets$ {\sc ColIndicator}$(A, B, C, V)$
\For{$(i,j)\in S_0\times T_0$}
    \State $C_{ij}\gets (AB)_{ij}$ 
\EndFor \label{proc:phase1_end}
\State $(\ind{IR}, S)\gets $ {\sc RowIndicator}$(A, B, C, V)$ \label{proc:phase2_start}
\State $(\ind{IC}, T)\gets$ {\sc ColIndicator}$(A, B, C, V)$
\While{$S_0\cap S\neq \emptyset$}  \label{step:row_indicator_loop_start}
    \State $i\gets\mbox{an element of $S_0\cap S$}$  \label{setp:retrieve_row}
    \State $T'\gets \{j'\in[n]\colon\, (AB)_{ij'}\neq C_{ij'} \}$   \Comment{Identify the nonzero columns}
    \For{$(i',j')\in [n]\times T'$}
        \State $(\ind{IR},S), (\ind{IC}, T)\gets ${\sc Update}$(i', j', A, B, C, V, (\ind{IR},S),  (\ind{IC}, T))$
        \State $C_{i'j'}\gets (AB)_{i'j'}$  \label{step:column_update}
    \EndFor \label{step:row_indicator_loop_end}
\EndWhile
\While{$T_0\cap T\neq\emptyset$}  
	\State $j\gets \mbox{an element of $T_0\cap T$}$ \label{setp:retrieve_col}
    \State $S'\gets \{i'\in[n]\colon\, (AB)_{i'j}\neq C_{i'j} \}$ \Comment{Identify the nonzero rows}
    \For{$(i',j')\in S'\times [n]$} 
        \State $(\ind{IR},S), (\ind{IC}, T)\gets ${\sc Update}$(i', j', A, B, C, V, (\ind{IR},S), (\ind{IC}, T))$
        \State $C_{i'j'}\gets (AB)_{i'j'}$  \label{step:row_update}
    \EndFor
\EndWhile
\end{algorithmic}
\end{algorithm}

For the correctness of Algorithm~\ref{alg:main}, first we observe that once a row or a column is detected, it remains detectable during the execution of the algorithm until all erroneous entries in the row or the column are corrected. This observation is an immediate consequence of Lemma~\ref{clm:1st_modification}.

\begin{lemma}
\label{clm:1st_modification}
    Before the execution of line~\ref{proc:phase2_start}, for $i\in S_0\cap S$ row $i$ of $AB-C$ contains at most $\sqrt{k}$ nonzeros in $M_0$. Similarly, for $j\in T_0\cap T$ column $j$ contains at most $\sqrt{k}$ nonzeros in $M_0$. 
\end{lemma}

\begin{proof}
    If row $i$ of $M_0$ contains no more than $\sqrt{k}$ nonzeros, the lemma holds immediately. Otherwise, since the entries that resides in a column containing at most $\sqrt{k}$ nonzeros were corrected at line~\ref{proc:phase1_end}, by the pigeonhole principle less than $\sqrt{k}$ entries remains in row $i$. A similar argument applies for $j\in T_0\cap T$.  
\qed\end{proof}

\begin{lemma}[Correctness of Algorithm~\ref{alg:main}]
\label{lem:correctness}
    All erroneous entries of $C$ are corrected during the execution of Algorithm~\ref{alg:main}.
\end{lemma}

\begin{proof}
    Suppose to the contrary that there exist $i$ and $j$ such that $C_{ij}$ is not corrected during the execution. Clearly, $i\notin S_0$ or $j\notin T_0$. If either $j\in T_0$ or $i\in S_0$, by Lemma~\ref{clm:1st_modification} $M_{ij}$ remains nonzero until $C_{ij}$ is corrected so it would be corrected at line~\ref{step:column_update} or line~\ref{step:row_update}. Thus we may assume that $i\notin S_0$ and $j\notin T_0$. Since $j\notin T_0$ and $M_{ij}\neq 0$, column $j$ of $M_0$ contains more than $\sqrt{k}$ nonzeros. By Remark~\ref{rmk:one_zero}, there is a nonzero entry $M_{i'j}$ such that $i'\in S_0$. If $i'$ is picked at line~\ref{step:row_indicator_loop_start}, then column $j$ of $C$ is corrected; otherwise all erroneous entries in row $i'$ of $C$ have to be corrected before $i'$ is picked during the loop. However, the correction of $C_{i'j}$ implies that all erroneous entries in column $j$ of $C$ are corrected, which leads to a contradiction. 
\qed\end{proof}

\medskip
Regarding the number of operations, keeping the row- and column-indicators up-to-date ensures that the process of identification is performed only if there is an erroneous entry, not corrected yet, in the corresponding row or column. We show this in Theorem~\ref{thm:time}.

\begin{theorem}
\label{thm:time}
    Given three $n\times n$ matrices of integers $A$, $B$, and $C$ with $C$ different from $AB$ by no more than $k$ entries, then all erroneous entries can be corrected by a deterministic algorithm using $O(k^2n+\sqrt{k}n^2)$ operations. In addition, the  values involved in running the algorithm are of $O(\alpha^2n^3)$, where $\alpha$ is the largest absolute value of an entry in $A$, $B$, and $C$. 
\end{theorem}

\begin{proof}
    We claim that Algorithm~\ref{alg:main} is the requested algorithm. The correctness of Algorithm~\ref{alg:main} follows from Lemma~\ref{lem:correctness}. For the number of operations, phase~1 (lines~\ref{proc:phase1_start} to~\ref{proc:phase1_end}) takes $O(k^2n)$ operations. For the remaining nonzero entries $M_{ij}$, clearly $i\notin S_0$ or $j\notin T_0$, and by Lemma~\ref{lem:detectability} each nonzero row or column identified in phase~2 contains at least $\sqrt{k}$ nonzeros in $M_0$. Thus, the number of distinct rows and columns to be corrected in phase~2 is at most $\sqrt{k}$. Since a row or a column is corrected right after being identified and the indicators are kept up-to-date, a row or a column would not be identified twice. Along with the fact that {\sc Update} takes  $O(n+\sqrt{k})$ operations, the total number of operations is of $O(\sqrt{k}n^2)$.

    \medskip
    Last, according to Bertrand's postulate, there exists a prime $p$ with $n<p<2n$. It follows that the value of an entry in the certificate $V$ is of $O(n)$. Consequently, the largest absolute value of an entry in $(AB-C)V$ is of $O(\alpha^2n^3)$. 
    \qed
\end{proof}

\medskip
As a final remark, in Algorithm~\ref{alg:main} we have to retrieve elements from $S_0\cap S$ at line~\ref{setp:retrieve_row}, which is equivalent to finding an element of $S_0$ and determining whether it is also an element of the present $S$. Since $S_0$ is static and every element of $S_0$ is processed exactly once, this  can be done efficiently by using an array of integers to represent $S_0$ and $S$. Each integer is a counter, recording the number of nonzeros in the corresponding row of $\ind{IR}_0$/$\ind{IR}$. The same implementation can be applied for $T_0$ and $T$, so the expected complexity can be achieved. In addition, the prime $p$ can be found in $O(n^2)$ time, e.g. using the sieve of Eratosthenes, and then the certificate $V$ can be constructed in $O(\sqrt{k}n)$ time.

\section{Concluding remarks}

In this paper, we investigate the problem of correcting matrix products over the ring of integers. We propose a simple deterministic $O(\sqrt{k}n^2+k^2n)$-time algorithm, where $n$ is the row- and column-dimensions, and $k$ is an upper bound on the number of erroneous entries. When $k<n$, the proposed algorithm is more efficient than the naive $\Theta(n^3)$-time algorithm to multiply two $n\times n$ matrices. For certain $k$, it also slightly improves the existing ones, e.g.~\cite{GkasieniecLLPT17,Kunnemann18,kutzkov}. In addition, the values involved during the execution are in a reasonable range, namely polynomial in $n$ and the input values, so the algorithm can be implemented with the requested time complexity on conventional computation models. 

\section*{Acknowledgement}
	The authors thank the anonymous reviewers, whose comments help improving the presentation, and also thank Wing-Kai Hon and Meng-Tsung~Tsai for the discussion regarding the computation models. This research was supported in part by the Ministry of Science and Technology of Taiwan under contract MOST grant 110-2221-E-003-003-MY3.

\bibliographystyle{plain}
\bibliography{ref}

\end{document}